\documentclass[11pt]{article}
\usepackage{ifpdf}
\usepackage[draft,uselot,hylinks,titlepage,full]{myarticle}

\usepackage{amsmath}
\usepackage{amssymb}
\usepackage{amsthm}
\usepackage{xspace}
\usepackage{color}
\newcommand{\dual}[1]{{{#1}^\perp}}
\newcommand{\dualk}[2]{{{#1}^\perp_{#2}}}

\newcommand{\ip}[1]{\langle #1 \rangle}

\DeclareMathOperator{\supp}{{\mathrm supp}}
\DeclareMathOperator{\wt}{{\mathrm wt}}
\DeclareMathOperator{\rate}{{\mathrm rate}}

\newcommand{\dist}[1]{\Delta \left({#1}\right)}

\renewcommand{\F}{{\mathbf{F}}}
\newcommand{\expref}[2]{\hyperref[#2]{#1~\ref{#2}}}

\newcommand{\secref}[1]{\hyperref[#1]{Section~\ref{#1}}}
\newcommand{\subsecref}[1]{\hyperref[#1]{Section~\ref{#1}}}

\newcommand{\thmref}[1]{\hyperref[#1]{Theorem~\ref*{#1}}}
\newcommand{\lemref}[1]{\hyperref[#1]{Lemma~\ref*{#1}}}
\newcommand{\corref}[1]{\hyperref[#1]{Corollary~\ref*{#1}}}
\newcommand{\clmref}[1]{\hyperref[#1]{Claim~\ref*{#1}}}
\newcommand{\propref}[1]{\hyperref[#1]{Proposition~\ref*{#1}}}

\newcommand{\algoref}[1]{\hyperref[#1]{Algorithm~\ref*{#1}}}

\theoremstyle{definition}
\newtheorem{fact}[theorem]{Fact}
\newcommand{\conjref}[1]{\hyperref[#1]{Conjecture~\ref*{#1}}}

\newcommand{\defnref}[1]{\hyperref[#1]{Definition~\ref*{#1}}}
\let\defref\defnref 

\newcommand{\egref}[1]{\hyperref[#1]{Example~\ref*{#1}}}

\newcommand{\remref}[1]{\hyperref[#1]{Remark~\ref*{#1}}}

\numberwithin{equation}{section}

\newcommand{\ie}{i.\,e.}
\newcommand{\eg}{e.\,g.}

\begin{document}

\title{A Combination of Testability and Decodability by Tensor Products}

\author{Michael Viderman\thanks{The research was partially supported by the European Community's Seventh Framework Programme (FP7/2007-2013) under grant agreement number 240258 and by
grant number 2006104 by the US-Israel Binational Science Foundation.} \\ Computer Science Department \\
Technion --- Israel Institute of Technology \\
Haifa 32000, Israel \\
{\tt viderman@cs.technion.ac.il}}

\maketitle

\begin{abstract}
Ben-Sasson and Sudan (RSA 2006) showed that repeated tensor products of linear codes with a very large distance are locally testable. Due to the requirement of a very large distance the associated tensor products could be applied only over sufficiently large fields. Then Meir (SICOMP 2009) used this result (as a black box) to present a combinatorial construction of locally testable codes that match best known parameters. As a consequence, this construction was obtained over sufficiently large fields.

In this paper we improve the result of Ben-Sasson and Sudan and show that for \emph{any} linear codes the associated tensor products are locally testable. Consequently, the construction of Meir can be taken over any field, including the binary field.

Moreover, a combination of our result with the result of Spielman (IEEE IT, 1996) implies a construction of linear codes (over any field) that combine the following properties:
\begin{itemize}
\item have constant rate and constant relative distance;
\item have blocklength $n$ and testable with $n^{\epsilon}$ queries, for any constant $\epsilon > 0$;
\item linear time encodable and linear-time decodable from a constant fraction of errors.
\end{itemize}

Furthermore, a combination of our result with the result of Guruswami et al. (STOC 2009) implies a similar corollary regarding the list-decodable codes.
\end{abstract}

\section{Introduction}\label{sec:intro}
Locally testable codes (LTCs) are error correcting codes that have a tester, which is a randomized
algorithm with oracle access to the received word $x$. The tester reads a sublinear amount of information from $x$
and based on this "local view" decides if $x \in C$ or not. It should accept codewords with probability one, and
reject words that are far (in Hamming distance) from the code with noticeable probability.

Such codes are of interest in computer science due to their numerous connections
to probabilistically checkable proofs (PCPs) and property testing (see the surveys \cite{Trevisan, Goldreich} for more information). By now several different constructions of LTCs are known including codes based on
low-degree polynomials over finite fields \cite{BLR93,ALMSS98}, constructions based on PCPs of proximity/assignment testers \cite{BGHSV04,DR06} and sparse random linear codes \cite{BV10,KS07, KopSar10}. In this paper we study a different family of LTC constructions, namely, \emph{tensor codes}. Given two linear error correcting codes $C \subseteq \F^{n_1}, R \subseteq \F^{n_2}$ over a finite field $\F$, we define
their \emph{tensor product} to be the subspace $R \otimes C \subseteq \F^{n_1\times n_2}$ consisting of $n_1\times n_2$ matrices $M$ with entries in $\F$ having the property that every row of $M$ is a codeword of $R$ and every
column of $M$ is a codeword of $C$. If $C=R$ we use $C^2$ to denote $C\otimes C$ and for $i>2$ define $C^i=C\otimes C^{i-1}$. Note that the blocklength of $C^i$ is $n_1^i$.

Recently, tensor products were used to construct new families of LTCs \cite{BS06,meir}, new families of list-decodable codes \cite{GopGR09}, to give an alternative proof (see \cite{Meir10}) for IP=PSPACE theorem of \cite{Shamir92,Shen92} etc.

Ben-Sasson and Sudan \cite{BS06} suggested to use tensor product codes as a means to construct LTCs combinatorially.
A natural hope would be to expect that given a code $C \subseteq \F^{n_1}$, whenever the task is to test whether a word $M \in \F^{n_1^2}$ is in $C^2$, the tester for $C^2$ can choose a random row (or column) of $M$; and if $M$ was far from $C^2$, its typical row/column is far from $C$ and hence can be tested on being in $C$. As was shown in \cite{nonrobust,nonrobustthree,nonrobusttwo} this approach fails in general and is known to work only for the base codes having some non-trivial properties \cite{DSW06,BV09,weaklysmooth}.

Nevertheless, Ben-Sasson and Sudan \cite{BS06} showed that taking the repeated tensor products of any code $C\subseteq \F^n$ with sufficiently large distance results in a locally testable code with sublinear query complexity. More formally, they showed \cite[Theorem 2.6]{BS06} that for every $m \geq 3$ if $\left(\frac{\dist{C}-1}{n}\right)^m \geq \frac{7}{8}$ then $C^m$ is locally testable using $n^2$ queries. Note that the blocklength of $C^m$ is $n^m$ and query complexity is $n^2$. Hence, for example, if $m = 10$ we obtain a code with blocklength $N = n^{10}$ and query complexity $N^{0.2}=n^2$, under assumption that $\left(\frac{\dist{C}-1}{n}\right)^{10} \geq \frac{7}{8}$.

Let us explain some issues that remained open. First of all, it was remained unclear if the assumption about a very large distance of the base codes is necessary. Moreover, the requirement on the distance of the base code is dependent on the number of tensor products ($m$) one should apply. Note that less query complexity (relatively to the blocklength) one should get more tensor product operations should be applied. Thus the requirement about the distance of the base code is increased when the number of queries one should get is decreased. We notice also that the larger distance implies the larger underlying field. It follows that this result can not provide (via tensor products) arbitrary low sublinear query complexity ($N^{\epsilon}$ for every constant $\epsilon > 0$) over a fixed field.

In this paper we ask the following question: is it possible to achieve a similar result to \cite{BS06} but with no requirements about the base codes at all. A positive result to this question might seem surprising since it would imply  that \emph{any} linear error-correcting code can be involved in the construction of LTCs via tensor products.

We give a positive answer on this question and show that no assumptions about the base codes (or underlying fields) are not needed. Our result does not make any assumptions about the base codes, and in particular we do not assume that the base codes involved in tensor products have very large distance and hence it holds over any fields. This contrasts with previous works on the combinatorial constructions of LTCs due to Ben-Sasson and Sudan \cite{BS06} and Meir \cite{meir} which required very large base-code distance implying large field size. The constructions of best known LTCs \cite{BS05, D07,meir} were obtained over the large fields (when finally, the field size can be decreased through code concatenation). Our improvement of \cite{BS06} implies that the construction of Meir \cite{meir} (which achieves LTCs of best known parameters) can be taken directly over any field (including the binary field). We think that this improvement has a non-negligible role since the LTCs construction of Meir is combinatorial and the combinatorial constructions of LTCs (or PCPs) should be independent, as much as possible, of the algebraic terms such as ``polynomials'', ``field size'', ``extension field'' etc.

Given the fact that error-correcting codes play an important role in a complexity theory, and in particular, in different iterative protocols, it might be helpful to develop a general scheme for constructing the error-correcting codes that combine several different properties. E.g., it might be helpful to have a high-rate codes which combine such properties as local testing, efficient encoding and decoding from a constant fraction of errors.

We show that a combination of our results with the results of \cite{Spielman,GopGR09} implies the construction of high-rate codes which are both testable with sublinear query complexity, linear-time encodable and efficiently decodable (or list-decodable) from the constant fraction of errors.

\paragraph{Organization of the paper.} In the following section we provide background regarding tensor codes and locally testable codes. In Section \ref{sec:main-results} we state our main results. We prove our
main theorem (\thmref{thm:main}) in Section \ref{sec:proofthmmain}. Finally,
in Section \ref{sec:aux} we prove our auxiliary statements.

\section{Preliminaries}\label{sec:definitions}

Throughout this paper, $\F$ is a finite field, $[n]$ denotes the set $\set{1,\ldots,n}$ and $\F^n$ denotes $\F^{[n]}$. All codes discussed in this paper will be a linear. Let $C\subseteq \F^n$ be a linear code over $\F$.

For $w \in \F^n$ let $\supp(w)=\{i|w_i \neq 0\}$, $|w|=|\supp(w)|$ and $\wt(w)=\frac{|w|}{n}$. We define the \emph{distance} between two words $x,y \in \F^n$ to be $\dist{x,y} = |\{i\ |\ x_i \neq y_i\}|$ and the relative
distance to be $\delta(x,y) = \frac{\dist{x,y}}{n}$. The distance of a code is defined by $\dist{C} = \min_{x\neq y \in C} \dist{x,y}$ and its the relative distance is denoted $\delta(C)=\frac{\dist{C}}{n}$. A $[n,k,d]_\F$-code is a $k$-dimensional subspace $C \subseteq \F^n$ of distance $d$. The rate of the code $C$ is defined by $\rate(C) = \frac{\dim(C)}{n}$. For $x \in \F^n$ and $C \subseteq \F^n$, let $\displaystyle \delta(x, C) = \delta_C(x)=\min_{y\in C}\left\{\delta(x,y) \right\}$ to denote the relative distance of $x$ from the code $C$. We note that $\displaystyle \dist{C} = \min_{c\in C \setminus \set{0}} \left\{ \wt(c) \right\}$. If $\delta(x, C) \geq \epsilon$ we say that $x$ is $\epsilon$-far from $C$ and otherwise $x$ is $\epsilon$-close to $C$. We let $\dim(C)$ denote the dimension of $C$. The vector inner product between $u=(u_1,u_2,\ldots,u_n)\in \F^n$ and $v=(v_1,v_2,\ldots,v_n)\in \F^n$ is defined to be
$\ip{u,v}=\sum_{i\in[n]} u_i \cdot v_i$. We let $\displaystyle \dual{C}=\set{u \in \F^n\ |\ \forall c \in C:\ \ip{u,c}=0}$ be the dual code of $C$ and $\displaystyle \dualk{C}{t} = \set{u \in \dual{C} \ |\ |u|=t}$. In a
similar way we define $\dualk{C}{\leq t} = \left\{u\in \dual{C} \ |\ |u|\leq t \right\}$. For $t\in \F^n$ and $T \subseteq \F^n$ we say that $t \perp T$ if $\ip{t, t'}=0$ for all $t' \in T$.

For $w \in F^n$ and $S = \left\{j_1,j_2,\dots,j_m \right\} \subseteq [n]$, where $j_1 < j_2 < \ldots < j_m$, we let $\displaystyle w|_{S} = (w_{j_1},\dots,w_{j_m})$ be the \emph{restriction} of $w$ to the subset $S$.
We let $\displaystyle C|_{S} = \set{c|_S\ |\ c\in C}$ denote the restriction of the code $C$ to the subset $S$.

\subsection{Tensor Product Codes}
The definitions appearing here are standard in the literature on tensor-based LTCs (\eg ~\cite{DSW06, BS06, meir, weaklysmooth, nonrobust}).

For $x\in \F^I$ and $y \in \F^J$ we let $x\otimes y$ denote the tensor product of $x$ and $y$ (\ie, the matrix $M$ with entries $M_{(i,j)} = x_i \cdot y_j$ where $(i,j) \in I \times J$). Let $R \subseteq \F^I$ and $C
\subseteq \F^J$ be linear codes. We define the tensor product code $R \otimes C$ to be the linear space spanned by words $r \otimes c \in \F^{I \times J}$ for $r \in R$ and $c \in C$. Some known facts regarding the tensor products (see \eg, \cite{DSW06}):
\begin{itemize}
\item The code $R \otimes C$ consists of all $I \times J$ matrices over $\F$ whose rows belong to $R$ and whose columns belong to $C$.

\item $\dim(R \otimes C)= \dim(R)\cdot \dim(C)$

\item $\rate(R \otimes C) = \rate(R) \cdot \rate(C)$

\item $\delta(R \otimes C)=\delta(R)\cdot\delta(C)$
\end{itemize}
We let $C^1 = C$ and $C^t = C^{t-1} \otimes C$ for $t > 1$. Note by this definition, $C^{2^0} = C$ and $C^{2^t} = C^{2^{t-1}} \otimes C^{2^{t-1}}$ for $t > 0$. We also notice that for a code $C \subseteq \F^n$ and $m \geq 1$ it holds that $\rate(C^m) = (\rate(C))^m$, $\delta(C^m) = (\delta(C))^m$ and the blocklength of $C^m$ is $n^m$.

The main drawback of the tensor product operation is that this operation strongly decreases the rate and the distance of the base codes. We refer the reader to \cite{meir} which showed how to use tensor products and avoid
the decrease in the distance and the strong decrease in the rate.

\subsection{Locally testable codes and Robustly Testable Codes}
A \emph{standard $q$-query tester} for a linear code $C \subseteq \F^n$ is a randomized algorithm that on the input word $w\in \F^n$ picks non-adaptively a subset $I \subseteq [n]$ such that $|I| \leq q$. Then $T$ reads all symbols of $w|_I$ and accepts if $w|_I \in C|_I$, and rejects otherwise (see \cite[Theorem 2]{3CNF}). Hence a $q$-query tester can be associated with a distribution over subsets $I \subseteq [n]$ such that $|I| \leq q$.

For purposes of composition we want to define a generalized tester (\defref{def:tester}) which does not make queries, but selects and returns a ``view'' (a subset $I \subseteq [n]$) which can be considered as a code by itself ($C|_I$).

\begin{definition}[Tester of $C$ and Test View]\label{def:tester}\label{def:locview}
A \emph{$q$-query tester} $\mathbf{D}$ is a distribution $\mathbf{D}$ over subsets $I \subseteq [n]$ such that $|I| \leq q$. Let $w \in \F^n$ (think of the task of testing whether $w \in C$) and let $I \subseteq [n]$ be a subset.
We call $w|_I$ the \emph{view} of a tester. If $w|_I \in C|_I$ we say that this view is \emph{consistent} with $C$, or when $C$ is clear from the context we simply say $w|_I$ is \emph{consistent}.
\end{definition}

When considering a tensor code $C^m \subseteq \F^{n^m}$, an associated tester will be a distribution over subsets $I \subseteq [n]^m$. Although the tester does not output $\acc$ or $\rej$, the way a
standard tester does, it can be converted to output $\acc, \rej$ as follows. Whenever the task is to test whether $w \in C$ and a subset $I \subseteq [n]$ is selected by the tester, the tester can output $\acc$ if $w|_{I}
\in C|_{I}$ and otherwise output $\rej$.

\begin{definition}[LTCs and strong LTCs]\label{def:lts}
A code $C\subseteq \F^n$ is a $(q,\epsilon,\delta)$-LTC if it has a $q$-query tester $\mathbf{D}$ such that for all $w\in \F^n$, if $\delta(w,C)\geq \delta$ we have $\displaystyle \pr{I \sim \mathbf{D}}{w|_{I} \notin
C|_{I}}\geq \epsilon$.

A code $C\subseteq \F^n$ is a $(q,\epsilon)$-strong LTC if it has a $q$-query tester $\mathbf{D}$ such that for all $w\in \F^n$, we have $\displaystyle \pr{I \sim \mathbf{D}}{w|_{I} \notin C|_{I}}\geq \epsilon\cdot
\delta(w,C)$.
\end{definition}

We notice that a $(q,\epsilon)$-strong LTC is a $(q,\epsilon \delta,\delta)$-LTC for every $\delta > 0$. Note that given a code $C \subseteq \F^n$, the subset $I\subseteq [n]$ uniquely defines $C|_I$. Moreover, the linearity of $C$ implies that $C|_I$ is a linear subspace of $\F^I$. In the rest of this section we formally
define the notion of \emph{robustness} (\defref{def:robust}) as was introduced in \cite{BS06}. To do that we start from the definition of \emph{local distance} (\defref{def:locdist}), which will be used in \defref{def:robust} and later in our proofs.

\begin{definition}[Local distance]\label{def:locdist}
Let $C$ be a code and $w|_I$ be the view on the coordinate set $I$ obtained from the word $w$. The \emph{local distance} of $w$ from $C$ with respect to $I$ (also called the $I$-distance of $w$ from $C$) is $\displaystyle \dist{w|_I,
C|_I} = \min_{c \in C} \left\{ \dist{w|_I,c|_I} \right\}$ and similarly the \emph{relative local distance} of $w$ from $C$ with respect to $I$ (relative $I$-distance of $w$ from $C$) is $\displaystyle \delta(w|_I, C|_I)
=\min_{c \in C} \left\{ \delta(w|_I,c|_I) \right\}$.
\end{definition}

Informally, robustness implies that if a word is far from the code then, on average, a test's view is far from any consistent view that can be accepted on the same coordinate set $I$. This notion was defined for LTCs
following an analogous definition for PCPs \cite{BGHSV04, D07}. We are ready to provide a general definition of robustness.

\begin{definition}[Robustness]\label{def:robust}
Given a tester (\ie, a distribution) $\mathbf{D}$ for the code $C \subseteq \F^n$, we let \[\rho^{\mathbf{D}}(w)=\ex{I \sim \mathbf{D}}{\delta(w|_I,C|_I)} \text{\ \ be the expected relative local distance of input $w$.}\] We
say that the tester $\mathbf{D}$ has robustness $\rho^{\mathbf{D}}(C)$ on the code $C$ if for every $w\in \F^n$ it holds that $\rho^{\mathbf{D}}(w) \geq \rho^{\mathbf{D}}(C) \cdot \delta_{C}(w)$.

Let $\set{C_n}_{n}$ be a family of codes where $C_n$ is of blocklength $n$ and $\mathbf{D_n}$ is a tester for $C_n$. A family of codes $\set{C_n}_{n}$ is {\em robustly testable} with respect to testers
$\{\mathbf{D_n}\}_{n}$ if there exists a constant $\alpha > 0$ such that for all $n$ we have $\rho^{\mathbf{D_n}}(C_n) \geq \alpha$.
\end{definition}

\section{Main Results}\label{sec:main-results}
The tester we consider in this paper is the plane tester (suggested in \cite{BS06}).

\begin{definition}[Plane Tester]\label{def:planes}
Let $m \geq 3$. Let $M \in \F^{n^m}$ be an input word and think of testing whether $M \in C^{n^m}$. The plane tester $\D$ picks (non-adaptively) a random $b\in [3]$ and random $i \in [n]$, and returns $(b,i)$-plane (the corresponding local view is $M|_{(b,i)}$).
Note that if $M$ is a candidate word to be in $C^m$ then $M|_{(b,i)}$ is a candidate word to be in $C^{m-1}$.
\end{definition}

Now we state our main technical theorem which says that the tensor product of any base code (with constant relative distance) is robustly testable. This extends the result of \cite{BS06} which showed that this claim holds for base codes with a very large distance.

\begin{theorem}[Main Theorem]\label{thm:main}
Let $C \subseteq \F^n$ be a linear code and $m \geq 3$. Let $\D$ be the plane tester for $C^m$. Then
\[\rho^{\D}(C^m) \geq \frac{(\delta(C))^{m}}{2m^2}.\]
\end{theorem}

The proof of Theorem \ref{thm:main} is postponed to Section \ref{sec:proofthmmain}.
Theorem \ref{thm:main} extends the main result of Ben-Sasson and Sudan \cite{BS06} since it implies that the $m$-wise  tensor product of linear codes is robust for any linear base codes with constant relative distance. In particular, the tensor product can be applied over any field, including the binary field. So, as explained in the introduction, the combinatorial construction of LTCs in \cite{meir} can be taken over any field (regardless of the field size).

Ben-Sasson and Sudan \cite{BS06} explained that plane testers can be composed and the robustness of the plane testers
implies the strong local testability. For the sake of completeness we state this claim formally in \corref{cor:test}, and provide a proof-sketch in Section \ref{sec:aux} (see \cite{BS06, BV09} for more information about composition of the testers).

\begin{corollary}\label{cor:test}
Let $C \subseteq \F^n$ be a linear code and $m \geq 3$ is a constant. Then
$C^m$ is a $(n^2, \alpha_m)$-strong LTC, where $\alpha_m > 0$ is a constant that depends only on $m$ and $\delta(C)$.
Note that the blocklength of $C^m$ is $n^m$.
\end{corollary}

\corref{cor:test} implies that any linear code can be used to define a locally testable code with sublinear query complexity. \clmref{clm:encod} shows that if a linear code $C$ is linear-time encodable then so is $C^i$ for any constant $i$. Later we will use this claim together with \corref{cor:test} to show \corref{cor:main}.

\begin{claim}\label{clm:encod}
Let $m \geq 1$ be a constant. If $C \subseteq \F^n$ is a linear-time encodable linear code then $C^m$ is linear-time encodable.
\end{claim}

The proof of \claimref{clm:encod} is postponed to Section \ref{sec:aux}. Now, we combine \corref{cor:test} and \clmref{clm:encod} to show a simple construction of strong LTCs with arbitrary small sublinear query complexity and arbitrary high rate from any linear code with sufficiently high rate.

\begin{corollary}\label{cor:main}
Let $C \subseteq \F^n$ be a linear code and let $m \geq 3$ be a constant. Then $C^m \subseteq \F^{n^m}$ is a $(n^2,\alpha_m)$-strong LTC, where $\alpha_m > 0$ is a constant that depends only on $m$ and $\delta(C)$.

In particular, for every $\epsilon > 0$, $m = \lceil\frac{1}{\epsilon}\rceil$, $N = n^m$ and $C\subseteq \F^n$ such that $\rate(C) \geq (1-\epsilon)^{1/m}$ we have $C^m \subseteq \F^N$ is a $(N^{\epsilon}, \alpha)$-strong LTC and $\rate(C^m) \geq 1-\epsilon$, where $\alpha > 0$ is a constant that depends only on $\epsilon$.
Moreover, if $C$ is a linear-time encodable then $C^m$ is a linear-time encodable.
\end{corollary}

Usually, in the areas of locally testable and locally decodable codes the main interest was given to the constant query complexity. Recently, Kopparty et al. \cite{KSY10} showed the construction of high-rate locally decodable codes with sublinear query complexity (see \cite{KSY10} for the motivation behind this range of parameters). Since then, the interest to the other range of parameters, and in particular, to sublinear query complexity was increased.

We would like to stress that \corref{cor:main} is quite powerful for this range of parameters (sublinear query complexity and high rate). First of all, there are different constructions
of linear-time encodable codes with constant rate and constant relative distance \cite{GI03,GI05, Spielman}, and them all can be involved to define high-rate LTCs with sublinear query complexity that are linear-time encodable.
The other advantage of such constructions is that the repeated tensor product of the base code is known to inherit some properties of the base codes besides local testability. E.g., Gopalan et al. \cite{GopGR09} showed that the tensor product operation preserves list-decodability properties. Furthermore, we know about nice constructions of error-correcting codes that can be efficiently encoded and decoded (list-decoded) from a constant fraction of errors (see e.g., \cite{GI03, GI05}).

In Section \ref{sec:uniquedec} we show how testing with sublinear query complexity can be combined with a linear time
encoding and decoding. Then, in Section \ref{sec:listdecod} we show that \corref{cor:main} can be combined with the result of \cite{GopGR09} to define asymptotically good codes that can be encodable in linear time, testable with sublinear query complexity and list-decodable in polynomial time.

\subsection{Linear-time decodable codes}\label{sec:uniquedec}
\propref{prop:lintimedecod} shows that the tensor product operation preserves the ``unique-neighbor'' decoding property.
In particular, if $C$ is a linear code that is linear time unique-neighbor decodable from a constant fraction of errors then so is $C^2$.  Hence this observation, together with a result of, e.g. \cite{Spielman}, can result in the construction of asymptotically good locally testable codes with sublinear query complexity that can be linear-time encoded and decoded to the closest neighbor after a constant fraction of errors.

\begin{proposition}\label{prop:lintimedecod}
Assume $C\subseteq \F^n$ is a linear code that is linear-time decodable from $\alpha \cdot n$ errors.
Then $C^2 = C \otimes C$ is a linear code that is linear-time decodable from $\frac{\alpha^2}{100} \cdot n^2$ errors.
\end{proposition}
\begin{proof}
Let $Dec_C$ be a linear-time decoder for the code $C$ that can correct any $\alpha \cdot n$ errors.
Note that in particular, $Dec_C$ correct any $\alpha \cdot n$ erasures in the linear time.
We define the linear-time decoder $Dec_{C^2}$ for the code $C^2$ that will correct any $\frac{\alpha^2}{100} \cdot n^2$ errors.

To do this, let $M \in \F^{n \times n}$ be an input word. The decoder $Dec_{C^2}$ will decode every row of $M$ using $Dec_C$ and every column of $M$ using $Dec_C$. Note that every entry of $M$ is contained in (exactly) one row and one column. Call the entry $(i,j)$ of $M$ an inconsistent if row decoding gives to $M|_{(i,j)}$ a different value from column decoding, and otherwise the entry is called consistent.

We call the row (column) of $M$ bad if it contains at least $\alpha n$ inconsistent entries.
Let $Bad_r$ be a number of bad rows and $Bad_c$ be a number of bad columns. It holds that $Bad_r \cdot \alpha n \leq \frac{\alpha^2 n^2}{100}$ and hence $Bad_r \leq \alpha n / 100$. Similarly, $Bad_c \leq \alpha n /100$.

The decoder $Dec_{C^2}$ removes all bad rows and bad columns that have at least $\alpha n/2$ inconsistent entries and obtains a large submatrix of size at least $(1-\alpha/100)n \times (1-\alpha/100)n$. It is easy to see that all consistent entries in the above large submatrix were decoded correctly.

In the last step, the decoder $Dec_{C^2}$ decodes, using $Dec_C$, every row of the large submatrix of $M$ (of size at least $(1-\alpha/100)n \times (1-\alpha/100)n$) from at most $\alpha n /100$ erasures and obtains a submatrix of size at least $(1-\alpha/100)n  \times n$. Now, it decodes every column of the submatrix to the full matrix.
It can be easily verified that the decoder $Dec_{C^2}$ obtains a correct codeword of $C^2$ and runs in linear time.
\end{proof}

While the results of \cite{Spielman} were improved, for our purpose (\corref{cor:testanddecod}) this result is sufficient.

\begin{theorem}[\cite{Spielman}]\label{thm:spielman}
There exists an (explicit) family of linear error correcting codes $C \subseteq \F_2^n$ such that
$\rate(C) = \Omega(1)$, $\delta(C) = \Omega(1)$, $C$ is a linear-time encodable and linear-time decodable from the constant fraction of errors.
\end{theorem}

A combination of \thmref{thm:spielman}, \propref{prop:lintimedecod} and \clmref{clm:encod} results in the following  corollary.

\begin{corollary}\label{cor:testanddecod}
For every constant $\epsilon > 0$ there exists an (explicit) family of linear error correcting codes $C \subseteq \F_2^N$ (obtained by tensor products on the codes from \thmref{thm:spielman}) that
\begin{itemize}
\item have constant rate and constant relative distance,
\item linear time encodable and linear time decodable from the constant fraction of errors,
\item are $(N^{\epsilon},\alpha)$-strong LTCs, where $\alpha = \alpha(\epsilon) > 0$ is a constant.
\end{itemize}
\end{corollary}

\subsection{Locally testable and list-decodable codes}\label{sec:listdecod}
In this section we recall some constructions of the list-decodable codes. We start by defining the list-decodable codes.

\begin{definition}[List-decodable codes]
A code $C$ is a $(\alpha,L)$-list decodable if for every word $w\in \F^n$, $\delta(w,C) \leq \alpha$ we have
$|\set{c\in C \ |\ \delta(c,w)\leq \alpha}| \leq L$. The code is said to be $(\alpha,L)$-list decodable in time $T$ if there exists algorithm which on the input $w\in \F^n$ such that $\delta(w,C) \leq \alpha$ outputs all codewords $c\in C$ such that $\delta(c,w)\leq \alpha$ (at most $L$ codewords).
\end{definition}

Guruswami et al. \cite{GopGR09} showed that the list-decodability is preserved in the tensor product operation. More formally, they showed the following theorem stated in \cite[Theorem 5.7]{GopGR09}.

\begin{theorem}[\cite{GopGR09}]\label{thmone}
Let $\F$ be a finite field and $q = |\F|$.
Given two linear codes $C_1, C_2 \subseteq \F^n$, for every $\epsilon > 0$, the number of codewords of $C_2 \otimes  C_1$ within distance $\eta^* = \min(\delta_1 \eta_2, \delta_2 \eta_1) - 3 \epsilon$ of any received word is bounded by
$l(C_2 \otimes C_1,\eta^*) \leq 4q^{\frac{1}{4\delta_1^2 \epsilon^2} \ln \frac{8l_1(\eta_1)}{\epsilon}\ln \frac{8l_2(\eta_2)}{\epsilon}}$.

Further, if $C_1$ and $C_2$ can be efficiently list decoded up to error rates $\eta_1, \eta_2$ and $C_2$ is a linear
code, then $C_2 \otimes C_1$ can be list decoded efficiently up to error rate $\eta^*$. Specifically, if $T$ denotes the time complexity of list decoding $C_1$ and $C_2$, then the running time of the list decoding algorithm
for $C_2 \otimes C_1$ is $O(4q^{\frac{1}{4\delta_1^2 \epsilon^2} \ln \frac{8l_1(\eta_1)}{\epsilon}\ln \frac{8l_2(\eta_2)}{\epsilon}}  \cdot T n_1 n_2)$.
\end{theorem}

Then, Gopalan et al. used Theorem \ref{thmone} to conclude the following theorem, appearing in \cite[Theorem 5.8]{GopGR09}.

\begin{theorem}[\cite{GopGR09}]\label{thmtwo}
Let $C$ be a linear code with distance $\delta$, list decodable up to an error rate $\eta$. For every $\delta > 0$, the $m$-wise tensor product code $C^m$ can be list decoded up to an error rate
$\delta^{m-1} \eta - \epsilon$ with a list size $\exp((O(\frac{\ln l(\eta)/\epsilon}{\epsilon^2}))^m)$.
Moreover, if $m \geq 1$ is constant and $C$ is polynomial-time list decodable then the running time of the list decoding algorithm for $C^m$ is polynomial (depending on $m$).
\end{theorem}

The next fact is known due to the several constructions of list-decodable codes.
\begin{fact}\label{factthree}
There exist linear error-correcting codes of constant rate and constant relative distance that can be encoded in linear time and list-decoded in polynomial time.
\end{fact}


We use the combination of \thmref{thmtwo}, Fact \ref{factthree}, \clmref{clm:encod}  and \corref{cor:main} to conclude the following corollary.

\begin{corollary}\label{cor:maintwo}
Let $\F$ be any field. For every constant $\epsilon > 0$ there exists a code $C' \subseteq \F^N$ such that
$C' = C^m$, where $C \subseteq \F^n$ is a linear code, $\rate(C) = \Omega(1)$, $\delta(C)=\Omega(1)$ and $C$ is $(\rho,L)$-list decodable in polynomial time.
\begin{itemize}
\item $C'$ is a $(N^{\epsilon},\alpha)$-strong LTC, where $\alpha = \alpha(\epsilon) > 0$ is a constant,
\item $C'$ is linear time encodable and list-decodable in polynomial time from the constant fraction of errors,
\item $\rate(C') \geq \Omega(1)$ and $\delta(C')= \Omega(1)$.
\end{itemize}
\end{corollary}

\section{Proof of Theorem \ref{thm:main}}\label{sec:proofthmmain}
Throughout this paper we assume that $C \subseteq \F^n$ is a linear code. We shall consider an $m$-wise tensor product, i.e., $C^m \subseteq \F^{n^m}$. Note that the blocklength of $C^m$ is $n^m$. Throughout this paper we assume that $m \geq 3$ and for the case of $m=2$ we refer a reader to \cite{BV09,weaklysmooth,DSW06,nonrobust,nonrobustthree,nonrobusttwo,nonrobusttwo}).
We start this section by defining the concepts of points, lines and planes (some of the terms were defined following \cite{BS06}).

\subsection{Preliminary notations: Points, Lines and Planes}\label{sec:planeslines}
A point in such a code can be associated with an $m$-tuple $(i_1,i_2,...,i_m)$ such that $i_j \in [n]$.
Next we define an axis parallel line, or shortly, a line which can be associated with a subset of points.
For $b\in [m]$ and $i\in [n]$ we say that $l$ is a $(b,(i_1,i_2,...,i_{b-1},i_{b+1},...,i_m))$-line if
\[l = \set{(i_1,i_2,...,i_{b-1},i,i_{b+1},...,i_m) \ |\ \text{ \ for all } j\in [m]\setminus \set{b} \text{ \ we have \ } i_j = i_j}.\]
Note that $(b,(i_1,i_2,...,i_{b-1},i_{b+1},...,i_m))$-line is parallel to the $b$-th axis.
A line $l$ contains a point $p$ if $p\in l$. Note that a $(b,(i_1,i_2,...,i_{b-1},i_{b+1},...,i_m))$-line contains a point $p=(j_1,i_2,...,j_m)$ if for all $k \in [m]\setminus \set{b}$ we have $i_k = j_k$. Two (different) lines intersects on the point $p$ if both lines contain the point $p$.

We say that $pl$ is a $(b,i)$-plane if
\[pl = \set{(i_1,i_2,...,i_m) \ | \ i_b=i  \text{ \ and for all } j\in [m]\setminus \set{b} \text{ \ we have \ } i_j \in [n]}.\]

A $(b,i)$-plane contains a point $p=(i_1,i_2,...,i_m)$ if $i_b = i$, i.e., $b$-th coordinate of the point is $i$. A $(b,i)$-plane contains a line $l$ if it contains all points of the line. We say that two (different) planes are intersected if both planes contain at least one common point. Note that two (different) planes: $(b_1,i_1)$-plane and $(b_2,i_2)$-plane are intersected iff $b_1 \neq b_2$, moreover, they are intersected on all points
$p=(i_1,\ldots,i_m)$ such that $i_1 = i_{b_1}$ and $i_2 = i_{b_2}$, i.e., are intersected on $n^{m-2}$ points.

Assume that $pl_1$ is a $(b_1,i_1)$-plane and $pl_2$ is a $(b_2,i_2)$-plane such that $b_1 < b_2$ (in particular $b_1 \neq b_2$). Let $pl_1 \cap pl_2 = \set{(i_1,\ldots,i_m) \ |\ i_{b_1}=i_1, i_{b_2}=i_2}$ be an intersection of two planes and $C^m|_{pl_1 \cap pl_2}$ be a code $C^m$ restricted to the points in $pl_1 \cap pl_2$. Note that
$\delta(C^m|_{pl_1 \cap pl_2}) = \delta(C^{m-2}) = \delta(C)^{m-2}$.

Given a word $M \in \F^{n^m}$, $b\in [m]$ and $i \in [n]$ we let $M_{(b,i)}$ be a restriction of $M$ to the $(b,i)$-plane, i.e., to all points of the plane. We say that $M_{(b,i)}$ is a $(b,i)$-plane of $M$.
Similarly, for the point $p=(i_1,\ldots,i_m)$ let $M|_p$ be a restriction of $M$ to the point $p$ and for the line $l$ we let $M|_l$ be a restriction of $M$ to the line $l$. We say that $M|_l$ is a line $l$ of $M$.

\subsection{The proof itself}
Let $M \in \F^{n^m}$ be an input word. We prove that $\rho^{\D}(M) \geq \frac{(\delta(C))^{m-1}}{2m^2} \delta(M,C^m)$.

For every plane $pl$ of $M$ let $r(pl)$ be the closest codeword of $C^{m-1}$ to $M|_{pl}$ (if there are more than one such codewords fix anyone arbitrarily). Intuitively, the plane $pl$ of $M$ ``thinks'' that the symbols of $M|_{pl}$ should be changed to $r(pl)$. In this sense every plane of $M$ has its own ``opinion''. Then we have
\begin{equation}\label{eq}
\rho^{\D}(M) = \ex{pl \sim \D}{\delta(M|_{pl},r(pl))}.
\end{equation}

We say that the $(b_1,i_1)$-plane and the $(b_2,i_2)$-plane disagree on the point $p=(i_1,\ldots,i_m)$ if $(b_1,i_1)$-plane and $(b_2,i_2)$-plane are intersected, both contain the point $p$ and $r(pl_1)|_{p} \neq r(pl_2)|_{p}$. We say that two planes disagree on the line $l$ if both planes are intersected, both contain the line $l$ and $r(pl_1)|_l \neq r(pl_2)|_l$.

Note that if $(b_1,i_1)$-plane $pl_1$ and $(b_2,i_2)$-plane $pl_2$ are intersected and disagree on at least one point then letting $reg = pl_1 \cap pl_2$ we have $r(pl_1)|_{reg} \neq r(pl_2)|_{reg}$ and moreover, $\delta(r(pl_1)|_{reg}, r(pl_2)|_{reg}) \geq (\delta(C))^{m-2}$. This is true since $r(pl_1)|_{reg}, r(pl_2)|_{reg} \in C^{m-2}$ are non-equal codewords of $C^{m-2}$ and $\delta(C^{m-2}) = (\delta(C))^{m-2}$.

Let $E \in \F_2^{n^m}$ be a binary matrix such that $E|_p =1$ if at least two planes disagree on the point $p$, and otherwise $E|_p =0$. For the point $p$ we say that the point is almost fixed if $E|_p = 0$ but $p$ is contained in some plane $pl$ such that $r(pl)|_p \neq M|_p$. Intuitively, a point $p$ is almost fixed if all planes containing this point agree on this point but ``think'' that its value in $M$ ($M|_p$) should be changed (to $r(pl)|_p$).

We let $ToFix = \set{p=(i_1,i_2,\ldots,i_m) \ |\ p \text{ is almost fixed}}$ and let $NumToFix = |ToFix|$.

\begin{proposition}\label{prop:boundrob}
It holds that $\rho^{\D}(M) \geq \frac{\wt(E)}{m} + \frac{NumToFix}{n^m}$.
\end{proposition}
\begin{proof}
Equation \ref{eq} says that $\rho^{\D}(M)$ is a relative distance of a typical plane of $M$ (which is a word in $\F^{n^{m-1}}$) from $C^{m-1}$. Note that for every point $p=(i_1,\ldots,i_m)$: if $E|_p \neq 0$ then $p\notin NeedToFix$. That means for every point $p$ at most one condition is satisfied: $E|_p \neq 0$ or $p\in NeedToFix$.

Note also that for every point $p \in NeedToFix$, for all planes $pl$ of $M$ we have $(M|_{pl})|_p \neq r(pl)|_p$. Now, every point $p$ is contained in $m$ different planes. Hence if $E|_p \neq 0$ then for at least one plane $pl$ (of $m$ planes containing the point $p$) it holds that $r(pl)|_{p} \neq M|_{p}$.

Hence a relative distance between a typical plane ($pl$) of $M$ and $r(pl)$ is at least $\frac{\wt(E)}{m} + \frac{NumToFix}{n^m}$.
\end{proof}

Next we define an important concept of ``heavy planes (lines)'' in the inconsistency matrix $E$. Intuitively, a heavy plane (line) of the matrix $E$ is a plane (line) where many inconsistencies occur, i.e., many non-zero symbols.

\begin{definition}[Heavy lines and planes]\label{def:heavyplanes}
A line $l$ of $E$ is called heavy if $|E|_{l}| \geq \delta(C) \cdot n$.
A plane $(b,i)$ of $E$ is called heavy if $|E|_{(b,i)}| \geq \frac{(\delta(C) \cdot n)^{m-1}}{2}$.
\end{definition}

Lemma \ref{lem:main} is our main observation in the proof of Theorem \ref{thm:main}. It says that any non-zero element of $E$ is located in some heavy plane of $E$.

\begin{lemma}[Main Lemma]\label{lem:main}
Let $p=(i_1,i_2,\ldots,i_m)$ be a point such that $E_p \neq 0$. Then $p$ is contained in some heavy plane of $E$.
\end{lemma}

The proof of Lemma \ref{lem:main} is postponed to Section \ref{sec:proofmainlem}.
\corref{cor:mainlem} shows that it is sufficient to remove at most $\wt(E) \cdot (\delta(C))^{m-1}/2 \cdot \frac{n}{2}$ planes from $E$ to get a zero submatrix.

\begin{corollary}\label{cor:mainlem}
There exists $S_1,\ldots,S_m \subseteq [n]$ such that $n-|S_1|+ n-|S_2|+ \ldots +n-|S_m| \leq \frac{2|E|}{(\delta(C)n)^{m-1}} \cdot m$ and letting $S = S_1 \times S_2 \times \ldots \times S_m$ we have $E|_{S}=0$.
\end{corollary}
\begin{proof}
Let $HeavyPlanes = \set{(b,i) \ |\ (b,i) \text{ \ is a heavy plane}}$ to be a subset of pairs associated with heavy planes. For $b \in [m]$ let $\overline{S_b} = \set{i \in [n] \ | \ (b,i) \in HeavyPlanes}$ and $S_b = [n] \setminus \overline{S_b}$.

We claim that $|HeavyPlanes| \leq \frac{2|E|}{(\delta(C)\cdot n)^{m-1}} \cdot m$. This is true since every heavy plane contains at least $\frac{(\delta(C)\cdot n)^{m-1}}{2}$ non-zero elements of $E$ and the total number of non-zero elements of $E$ is $|E|$. Furthermore, every non-zero element of $E$ is contained in at most $m$ (heavy) planes. Thus $n-|S_1|+ n-|S_2|+ \ldots +n-|S_m| = \sum_{b\in [m]} |\overline{S_b}| \leq \frac{2|E|}{(\delta(C)n)^{m-1}} \cdot m$.

Now, note that Lemma \ref{lem:main} implies that for every point $p=(i_1,i_2,\ldots,i_m)$ such that $E|_{p} \neq 0$ is contained in some heavy plane, i.e., in some plane from $HeavyPlanes$. Hence if all heavy planes are removed from $E$ we obtain a zero submatrix. So, it follows that $E|_{S}=0$.
\end{proof}

\propref{prop:decod} says that if after removing small fraction of planes from $M$ we obtain a submatrix that is close to the legal submatrix then $M$ is close to $C^{m}$.

\begin{proposition}\label{prop:decod}
Let $S_1, S_2, ..., S_m \subseteq [n]$ be such that $n - |S_1| + n - |S_2| + ... + n-|S_m| \leq \tau n < \delta(C) \cdot n$ and let $S = S_1 \times S_2 \times \ldots \times S_m$. Let $C' = C|_{S_1} \otimes C|_{S_2} \otimes \ldots \otimes C|_{S_m}$. Recall that $M|_S$ is a submatrix of $M$ obtained by removing at most $\tau$-fraction of planes. Assume that $\dist{M|_S,C'} \leq \alpha \cdot n^m$. Then $\delta(M, C^m) \leq \tau + \alpha$.
\end{proposition}

The proof of Proposition \ref{prop:decod} appears in Section \ref{sec:proofprop}. Let us prove Theorem \ref{thm:main}.

\begin{proof}[Proof of Theorem \ref{thm:main}]
By Proposition \ref{prop:boundrob} we have $\rho^{\D}(M) \geq \frac{\wt(E)}{m} + \frac{NumToFix}{n^m}$.
If $\wt(E) \geq \frac{(\delta(C))^{m}}{2 m}$ then we are done. Otherwise, assume that $\wt(E) < \frac{(\delta(C))^{m}}{2 m}$.

Corollary \ref{cor:mainlem} implies that it is sufficient to remove at most
$\frac{2|E|}{(\delta(C)n)^{m-1}} \cdot m < \delta(C) \cdot n$ planes from $E$ to get a zero submatrix.
Proposition \ref{prop:decod} implies that $\delta(M,C^m) \leq \frac{2\wt(E)}{(\delta(C))^{m-1}} \cdot m + \frac{NumToFix}{n^m}$.

Let $\beta = \frac{2m^2}{(\delta(C))^{m-1}}$. Then $\rho^{\D}(M) \cdot \beta \geq (\frac{\wt(E)}{m} + \frac{NumToFix}{n^m}) \cdot \beta \geq \delta(M,C^m)$ and $\rho^{\D}(M) \geq \frac{(\delta(C))^{m-1}}{2m^2} \cdot \delta(M,C^m)$.
\end{proof}

\subsubsection{Proof of Main Lemma \ref{lem:main}}\label{sec:proofmainlem}
In this section we prove Lemma \ref{lem:main}.

\begin{proof}[Proof od Main Lemma \ref{lem:main}]
By definition of $E$ we know that there are (at least) two planes that disagree on the point $p$.
Assume without loss of generality (symmetry) that the planes $p_1 = (1,i_1)$ and $p_2 = (2,i_2)$ disagree on the point $p$. We will prove that either $p_1$ is a heavy plane or $p_2$ is a heavy plane.

Consider the intersection of $p_1$ and $p_2$, i.e., $reg=pl_1 \cap pl_2 = \set{(i_1,i_2, j_3,j_4,\ldots,j_m) \ |\ j_k \in [n]}$. Note that $p\in reg$. Let $ln$ be a line, which is parallel to the axis $3$ and contains a point $p$ (recall that $m \geq 3$). Then the planes $p_1$ and $p_2$ disagree on this line (since they disagree on the point $p$ contained in the line $ln$), i.e., $r(p_1)|_{ln} \neq r(p_2)|_{ln}$. But $r(p_1)|_{ln}, r(p_2)|_{ln} \in C$ by definition. This implies that $\dist{r(p_1)|_{ln}, r(p_2)|_{ln}}\geq \delta(C) \cdot n$, i.e., for at least $\delta(C) \cdot n$ points $p\in ln$ it holds that $r(p_1)|_{p} \neq r(p_2)|_{p}$.

Let $BadPoints = \set{p \in ln \ |\ p_1\text{ and }p_2 \text{ \ disagree on \ }p}$. Note that $|BadPoints| \geq \delta(C) \cdot n$. Note that $|BadPoints| \geq \delta(C) \cdot n$.
Let $BadPlanes = \set{(3,i)-\text{plane} \ |\  i\in [n], \exists p \in BadPoints \text{ \  s.t. \ }\ p\in (3,i)-\text{plane}}$.
Note that $|BadPlains| \geq \delta(C) \cdot n$.

We claim that for every $plane \in BadPlanes$ we have that either $plane$ disagrees with $p_1$ on some point $p \in BadPoints$ or with $p_2$ on some point $p \in BadPoints$.
Hence at least one of $p_1,p_2$ disagrees with at least $\frac{1}{2} \cdot |BadPlanes| \geq \frac{1}{2} \cdot \delta(C) n$ planes from $BadPlanes$. Without loss of generality assume that $p_1$ disagrees with at least $\frac{1}{2} \cdot \delta(C) \cdot n$ planes from $BadPlanes$.

Let $BadPlanes_{p_1} = \set{plane \in BadPlanes \ |\ plane \text{ \ disagrees with \ }p_1}$. Note that all planes from $BadPlanes$ are non-intersecting and thus all planes from $BadPlanes_{p_1}$ are non-intersecting. Every plane $pl \in BadPlanes_{p_1}$ disagrees with the plane $p_1$ on some point and hence disagree on at least $(\delta(C) n)^{m-2}$ points in their intersection region ($pl \cap p_1$) since $r(pl)|_{pl \cap p_1} \neq r(p_1)|_{pl \cap p_1} \in C^{m-2}$.

Let $total = \set{p=(i_1,j_2,\ldots,j_m) \ |\ \exists plane \in BadPlanes_{p_1} \text{\ s.t.\ } p \in p_1 \cap plane,  r(plane)|_p \neq r(p_1)|_p}$. We have $|total| \geq (\delta(C) n)^{m-2} \cdot \frac{\delta(C)\cdot n}{2} = \frac{(\delta(C) \cdot n)^{m-1}}{2}$ since every intersection region (as above) contains at least $(\delta(C) n)^{m-2}$ inconsistency points and there are at least $\frac{1}{2} \cdot \delta(C) \cdot n$ such regions.
We stress that we do not count any inconsistency point more than once, since the planes in $BadPlanes_{p_1}$ are non-intersecting.

Hence the plane $p_1$ disagree with other planes in at least $\frac{(\delta(C)\cdot n)^{m-1}}{2}$ points (on the plane). Thus $E|_{p_1}$ has at least $\frac{(\delta(C)\cdot n)^{m-1}}{2}$ non-zero symbols.
We conclude that $p_1$ is a heavy plane of $E$ and the point $p$ is contained in the plane $p_1$.
\end{proof}

\subsubsection{Proof of Proposition \ref{prop:decod}}\label{sec:proofprop}
In this section we prove Proposition \ref{prop:decod}.

\begin{proof}[Proof of Proposition \ref{prop:decod}]
Note that for every $i\in [n]$ we have $|S_i|> n - \delta(C) \cdot n$.
The following simple claim was proven in \cite[Proposition 3.1]{BS06}. For the sake of completeness we provide its proof.

Every codeword $c'$ of $C'$ can be extended to a unique codeword $c$ of $C^m$.
To see this note that the projection of $C$ to $C|_{S_i}$ is bijective. It is surjective because it is a projection, and it is injective because $|S_i| > n - \dist{C}$. So, the projection of $C$ to $C'$ is bijection, because both codes are of dimension $(\dim(C))^m$. Thus, every word in $C'$ has a unique preimage in $C$.

We turn to prove Proposition \ref{prop:decod}.
We know that $M$ can be modified in at most $\alpha$-fraction of points $p \in S$ to get
$M|_S \in C'$. Then, by the claim above, $M$ can be modified (outside the submatrix $M|_S$) to a codeword of $C^{m}$ by changing at most $\tau$-fraction of symbols (since all symbols outside the submatrix $M|_S$ are at most $\tau$-fraction of all symbols). We conclude that $\delta(M, C^m) \leq \tau + \alpha$.
\end{proof}

\section{Proofs of Auxiliaries Corollaries and Claims}\label{sec:aux}

In this section we first proof \corref{cor:test}.

\begin{proofsk}
For $i \geq 3$ let $\D_i$ be the plane tester for the code $C^i$. Note that the tester $\D_m$ returns a local view that is a candidate to be in the code $C^{m-1}$. Hence $\D_{m-1}$ can be invoked on the local view of $\D_m$, etc.
So, the testers $\D_m, \D_{m-1},\ldots, \D_3$ can be composed to result in a $n^2$-query tester $\D_{comp}$ for the code $C^m$.

The robustness of the composed tester will be $\rho^{\D_{comp}}(C^m) \geq \rho^{\D_{m}}(C^m) \cdot \rho^{\D_{m-1}}(C^{m-1}) \cdot \ldots \cdot \rho^{\D_{3}}(C^3)$. To see this let $w\in \F^{n^m}$ be a word such that
$\delta(w,C^m) = \delta$. Then the local view of the tester $\D_m$ is expected to be $\rho^{\D_{m}}(C^m) \cdot \delta$ far from $C^{m-1}$. When $\D_{m-1}$ will be invoked, its local view will be $\rho^{\D_{m}}(C^m) \cdot \rho^{\D_{m-1}}(C^{m-1}) \cdot \delta$ far from $C^{m-2}$, etc.
Finally, the local view of $\D_3$ will be $(\rho^{\D_{m}}(C^m) \cdot \rho^{\D_{m-1}}(C^{m-1}) \cdot \ldots \cdot \rho^{\D_{3}}(C^3)) \cdot \delta$ far from $C^2$.

Theorem \ref{thm:main} says that $\rho^{\D_m}(C^m) \geq \frac{(\delta(C))^{m}}{2m^2}$.
Hence for constant $m\geq 3$ it holds that $\rho^{\D_{comp}}(C^m) > 0$ is a constant that depends only on $\delta(C)$ and $m$.

Now, let $\alpha_m = \rho^{\D_{comp}}(C^m)$ and note that the query complexity of $\D_{comp}$ is $n^2$. Let $M \in \F^{n^m}$ such that $\delta(M,C^m) = \delta$ and think of testing whether $M$ in $C^m$. We argue that $\pr{I \sim \D_{comp}}{M|_I \in C^m|_I} \geq \alpha_m$. This is true since otherwise, $\pr{I \sim \D_{comp}}{M|_I \notin C^m|_I} < \alpha_m$, and then the robustness $\rho^{\D_{comp}}(C^m) < \alpha_m \cdot 1$ (even assuming that whenever $M|_I \notin C^m|_I$ we have $\delta(M|_I, C^m|_I)=1$). Contradiction.

This proves that $C^m$ is a $(n^2,\alpha_m)$-strong LTC.
\end{proofsk}

Now we prove \clmref{clm:encod}.

\begin{proof}[Proof of Claim \ref{clm:encod}]
Let $k = \dim(C)$. Let $E_C$ be an encoder for the code $C$, which receives a message $x \in \F^k$ and outputs a codeword $E_C(x) \in C$ such that $C= \set{E_C(x) \ |\ x\in \F^k}$. Assume that $E_C$ has running time $T = O(k)$.
Note that this implies that $n \leq T = O(k)$ since the blocklength can not exceed the running time of the encoder.

For every $i \geq 1$ we define $E_{C^i}$ to be the encoder for $C^i$, i.e., $C^i = \set{E_{C^i}(x) \ |\ x\in \F^{k^i}}$. We will argue that the running time of $E_{C^i}$ is $i \cdot n^{i-1} \cdot T$.
Since $n \leq T = O(k)$ we will conclude that for any constant $i \geq 1$ the running time of $E_{C^i}$ is linear (in $k^i$).

We prove the claim by induction on $i$. The encoder $E_C = E_{C^1}$ was defined and its running time is $T = 1 \cdot n^{1-1} \cdot T$. Assume that we defined the encoder $E_{C^{i-1}}$ for the code $C^{i-1}$ and its running time is $(i-1) \cdot n^{(i-1)-1} \cdot T$.

Let us define the encoder $E_{C^i}$ for the code $C^i$. Note that the code $C^i$ has message length $k^i$ and its blocklength is $n^i$. Hence the message $x \in \F^{k^i}$ can be viewed as a matrix $k \times k^{i-1}$ and the encoder
$E_{C^i}$ will first encode (by the encoder $E_{C^{i-1}}$) every row of the matrix, obtaining the matrix $k \times n^{i-1}$. Then $E_{C^i}$ will encode every column of the obtained matrix to get a codeword of $C^i$.
The running time of the encoder $C^i$ is $k \cdot ((i-1) \cdot n^{i-2} T) + n^{i-1} T \leq  ((i-1) \cdot n^{i-1} T) + n^{i-1} T = i \cdot n^{i-1} \cdot T$.
\end{proof}

\subsection*{Acknowledgements}
We thank Eli Ben-Sasson and Or Meir for helpful discussions. The author thanks also Eli Ben-Sasson for many invaluable discussions regarding the ``robustness'' concept.

\bibliography{ImprovedTensors}
\bibliographystyle{plain}

\end{document}